\documentclass[12pt]{amsart}

\usepackage{amsmath,amsthm,amssymb,amsfonts,verbatim}
\usepackage[colorlinks,
            linkcolor=blue,
            anchorcolor=blue,
            citecolor=blue
            ]{hyperref}
\usepackage[hmargin=1.2in,vmargin=1.2in]{geometry}

\title[Codes via rational functions]{A new construction of nonlinear codes via rational function fields}
\author{Lingfei Jin}\address{Shanghai Key Laboratory of Intelligent Information Processing, School of Computer Science, Fudan University, Shanghai 200433, China.} \email{lfjin@fudan.edu.cn}
\author{Liming Ma}\address{School of Mathematical Sciences, Yangzhou University, Yangzhou China
225002}\email{lmma@yzu.edu.cn}
\author{Chaoping Xing} \address{School of Electronics, Information and Electric Engineering, Shanghai Jiao Tong University,
China 200240}\email{xingcp@sjtu.edu.cn}
\date{}

\newtheorem{lemma}{Lemma}[section]
\newtheorem{theorem}[lemma]{Theorem}
\newtheorem{cor}[lemma]{Corollary}

\newtheorem{prop}[lemma]{Proposition}
\newtheorem{ex}[lemma]{Example}

\theoremstyle{remark}
\newtheorem{rmk}{Remark}

\renewcommand{\epsilon}{\varepsilon}
\renewcommand{\le}{\leqslant}
\renewcommand{\ge}{\geqslant}

\def\Im{\mathrm{Im}}

\newcommand{\vnote}[1]{}

\def\ZZ{\mathbb{Z}}
\def\PP{\mathbb{P}}

\def\F{\mathbb{F}}

\def \mL {\mathcal{L}}

\def \Xi {{X^{[i]}}}

\newcommand{\Ga}{\alpha}

\def \bc {{\bf c}}

\def\supp {{\rm supp }}

\begin{document}
\maketitle

\begin{abstract}
It is well known that constructing codes with good parameters is one of the most important and fundamental problems in coding theory.
Though a great many of good codes have been produced, most of them are defined over alphabets of sizes equal to prime powers.
In this paper, we provide a new explicit construction of $(q+1)$-ary nonlinear codes via rational function fields, where $q$ is a prime power.
Our codes are constructed by evaluations of rational functions at all the rational places (including the place of ``infinity") of the rational function field.
Compared to the rational algebraic geometry codes, the main difference is that we allow rational functions to be evaluated at pole places.
After evaluating rational functions from a union of Riemann-Roch spaces, we obtain a family of nonlinear codes with length $q+1$ over the alphabet $\mathbb{F}_{q}\cup \{\infty\}$.
As a result, our codes have reasonable parameters as they are very close to the Singleton bound.
Furthermore, our codes have better parameters than those obtained from MDS codes via code alphabet restriction or extension.
\end{abstract}

\section{Introduction}

Since the birth of error-correcting codes, constructing codes with good parameters has become one of the most important and fundamental problems in coding theory.
For a $q$-ary code of length $n$, size $M$ and minimum distance $d$, we usually denote it by an $(n,M,d)$ code. When the length $n$ is fixed, the size $M$ is a measure of the efficiency of the code and the minimum distance $d$ represents the error correcting capability. Therefore, we usually hope the size $M$ and minimum distance $d$ to be as large as possible for given $n$ and $q$. However, there are several bounds on the largest possible value of $M$. One of the well-known bounds is the so called Singleton bound which says that $M\le q^{n-d+1}$. A linear code achieving this bound is called a maximum distance separable (MDS) code.

Many efforts have been devoted to the construction of good codes. In particular, linear codes have received great attention, such as Reed-Solomon (RS) codes, BCH codes, cyclic codes and so on, since they have good structures and many practical advantages. However, for given alphabet size $q$, length and minimum distance, the size of a nonlinear code may not be achieved by any linear codes. Indeed, there are some examples showing that linear codes do not exist for some parameters that nonlinear codes can have.
For example, there are no binary linear codes of parameters $[16,8,6]$. On the other hand, the Nordstorm-Robinson code \cite{LX04} is a binary nonlinear code with parameters $(16,2^8,6)$. Therefore, it is also of interest to provide explicit constructions of nonlinear codes.
Though a large number of nonlinear codes have been constructed, most of them are $q$-ary codes where $q$ is a prime power. The existing methods mainly consider nonlinear codes over finite fields. Less is known for constructions of $q$-ary codes, where $q$ is not a prime power, except for a very few results. Few examples, some nonlinear codes over  $\ZZ_6$, $\ZZ_{10}$ or $\ZZ_{12}$ were given with certain properties \cite{GH05,H00, HM11}.

In this paper, we focus on a construction of $(q+1)$-ary codes with $q$ being a prime power.
To better understand the idea of this paper, here we give a high-level description of our techniques.
Recall that a generalized Reed-Solomon code is constructed via evaluations of polynomials at $n$ distinct elements of $\F_q$ ($n\le q$).
Thus the length of a generalized Reed-Solomon code is upper bounded by $q$. If one includes the place of infinity, then we can obtain an extended Reed-Solomon code of length up to $q+1$.
For both generalized Reed-Solomon codes and  extended Reed-Solomon codes, the evaluations of polynomials still belong to $\F_q$. Hence, the codes have the alphabet size $q$.
Our idea is to extend polynomials to rational functions, i.e., consider evaluations of rational functions at all the rational places of the rational function field.
As a result, we produce a code of length $q+1$ over the code alphabet $\F_{q}\cup \{\infty\}$.
To estimate the minimum distance, we have to control the degrees of numerator and denominator of a rational function. This constraint affects the size of the code. Thus, we have to choose suitable rational functions to make  good trade-off between the minimum distance and the size of the code.

This paper is organized as follows. In Section \ref{sec:2}, we provide some background on the rational function field and coding theory. In Section \ref{sec:3}, we give an explicit construction of nonlinear codes from the rational function field. Numerical examples and comparison are given in Section \ref{sec:4}.

\section{Preliminaries}\label{sec:2}
In this section, we present some preliminaries on the theory of the rational function field, the Riemann-Roch space and coding theory.

\subsection{The rational function field}
Let us introduce some basic notations and facts of the rational function field. The reader may refer to \cite{St09} for more details.

Let $\F_q$ denote the finite field with $q$ elements. Denote by $F$ the rational function field $\F_q(x)$, where $x$ is a transcendental element over $\F_q$.
Every finite place $P$ of $F$ corresponds to a monic irreducible polynomial $p(x)\in\F_q[x]$ and its degree is equal to the degree of corresponding polynomial.
There is an infinite place of $F$ with degree one which is the unique zero of $1/x$ and denoted by $P_{\infty}$.
The set of places of $F$ is denoted by $\PP_F$. The place of degree one is called rational.
In fact, there are exactly $q+1$ rational places for the rational function field over $\F_q$, i.e., the place $P_{x-\Ga}$ for each $\Ga\in \F_q$ and the infinite place $P_{\infty}$.
Usually, we denote $P_{x-\Ga}$ by $P_{\Ga}$ for short.
Let $\Sigma$ denote the set $\F_q \cup \{\infty\}$. Then the set of all the rational places of $F$ can be identified with $\Sigma$.

Let $P$ be a rational place of $F$ and let $\mathcal{O}_P$ be the valuation ring with respect to $P$. For $f\in \mathcal{O}_P$,  we define $f(P)\in \mathcal{O}_P/P=\F_q$ to be the residue class of $f$ modulo $P$; otherwise for $f\in F\setminus \mathcal{O}_P$, we define $f(P)=\infty$.
In particular, if $f(x)=g(x)/h(x)\in \F_q(x)$ with relatively prime polynomials $g(x)=a_nx^n+\cdots+a_0$ and $h(x)=b_mx^m+\cdots+b_0$ with $a_n b_m\neq 0$, then the residue class map can be determined as follows $$f(P_\Ga)=\begin{cases} g(\Ga)/h(\Ga) & \text{ if } h(\Ga)\neq 0, \\ \infty & \text{ if } h(\Ga)=0\end{cases} $$
for any $\Ga\in\F_q$  and $$f(P_{\infty})=\begin{cases} a_n/b_m & \text{ if } n=m, \\ 0 & \text{ if } n<m,\\ \infty & \text{ if } n>m.\end{cases}$$

A divisor $G$ of $F$ is a formal sum $G=\sum_{P\in \PP_F} n_PP$ with only finitely many nonzero integers $n_P$.
The support of $G$ is defined as the set of places with nonzero coefficients in $G$.
Let $\nu_P$ be the normalized discrete valuation of $P$. For a nonzero element $f\in F$, the zero divisor of $f$ is defined by $(f)_0=\sum_{P\in \PP_F, \nu_{P}(f)>0} \nu_P(f)P,$
and the pole divisor of $f$ is defined by $(f)_\infty=\sum_{P\in \PP_F, \nu_P(f)<0} -\nu_P(f)P.$
The principal divisor of $f$ is given by $$(f)=(f)_0-(f)_\infty=\sum_{P\in \PP_F} \nu_P(f)P.$$
For two divisors $G=\sum_{P\in \PP_F} n_PP$ and $D=\sum_{P\in \PP_F} m_PP$, we define the union and intersection of $G$ and $D$ respectively as follows
\[G\vee D:=\sum_{P\in\PP_F} \max\{n_P,m_P\}P,\qquad G\wedge D:=\sum_{P\in\PP_F} \min\{n_P, m_P\}P.\]
The degree of $G$ is defined by $\deg(G)=\sum_{P\in \PP_F}n_P\deg(P)$. It is clear that
\begin{equation*}
\deg(G\wedge D) +  \deg(G\vee D)= \deg(G)+\deg(D).
\end{equation*}

\subsection{The Riemann-Roch space}
For a divisor $G$ of the rational function field $F/\F_q$, we define the Riemann-Roch space
\[\mL(G):=\{u\in F^*:\; (u)+G\ge 0\}\cup\{0\}.\]
From the Riemann-Roch theorem \cite[Theorem 1.5.17]{St09}, $\mL(G)$ is a vector space of dimension $\deg(G)+1$ over $\F_q$ for any divisor of nonnegative degree.
For example if $G=mP_\infty$ with $m>0$, then $\mL(G)$ is an $(m+1)$-dimensional vector space of polynomials of degree at most $m$.
It is straightforward to verify that
\begin{equation*}
\mL(G) \cap \mL(H) = \mL(G\wedge H) \ \makebox{and} \ \mL(G)+\mL(H) \subseteq \mL(G\vee H)
\end{equation*}
for any two divisors $G$ and $H$.
Furthermore, the following lemmas  will be very useful to determine the minimum distance of our codes constructed in the next section.

\begin{lemma}\label{lem:2.1}
Let $f$ be a nonzero function in $F$ with $(f)_\infty=G$.  Then for any $\Ga\in  \Sigma$,  $f(P_\Ga)=\infty$ if and only if $P_\Ga \in \supp(G)$.
\end{lemma}
\begin{proof}
It is easy to verify that
$$f(P_\Ga)=\infty \Leftrightarrow f\in F\setminus \mathcal{O}_{P_\Ga}   \Leftrightarrow  \nu_{P_\Ga}(f)\le -1  \Leftrightarrow  P_\Ga \in \supp(G)$$
from the definition of pole divisors.
\end{proof}

\begin{lemma}\label{lem:2.2}
Let $f$ be a nonzero function in $F$ with $(f)_\infty=G$.  Then for any $\Ga\in  \Sigma$ with $P_\Ga\not\in \supp(G)$,  $f(P_\Ga)=0$ if and only if $f\in\mL(G-P_\Ga)$.
\end{lemma}
\begin{proof}
It is easy to see that
$$f(P_\Ga)=0 \Leftrightarrow f\in P_\Ga   \Leftrightarrow  \nu_{P_\Ga}(f)\ge 1  \Leftrightarrow  (f)+G-P_\Ga\ge 0  \Leftrightarrow f\in\mL(G-P_\Ga)$$
from the definition of Riemann-Roch spaces.
\end{proof}

\begin{lemma}\label{lem:2.3}
Let $f_1,f_2$ be two nonzero functions in $F$ with pole divisors $(f_i)_\infty=G_i$ for $i=1,2$.
If we have $f_1(P_\Ga)=f_2(P_\Ga)$ for $\Ga\in  \Sigma$, then  $f_1-f_2\in\mL(G_1+ G_2-P_\Ga)$.
\end{lemma}

\begin{proof}
{\bf Case 1: } If $f_1(P_\Ga)=f_2(P_\Ga)\in\F_q$ for $\Ga\in\F_q$, then  $P_\Ga\not\in \supp(G_1\vee G_2)$. 
In this case, we have $(f_1-f_2)(P_\Ga)=f_1(P_\Ga)-f_2(P_\Ga)=0$. By Lemma \ref{lem:2.2}, we have $f_1-f_2\in\mL(G-P_\Ga)$, where $G=(f_1-f_2)_\infty$. As $G=(f_1-f_2)_\infty\le (f_1)_\infty+(f_2)_\infty= G_1+G_2$ from the strict triangle inequality \cite[Lemma 1.1.11]{St09}, the desired result follows.

{\bf Case 2: } If $f_1(P_\Ga)=f_2(P_\Ga)\in\F_q$ for $\Ga=\infty$, then $P_\infty\not\in \supp(G_1\vee G_2)$.
In this case, we may assume that $f_i=g_i/h_i$ with $g_i,h_i\in\F_q[x]$, where both $h_1$ and $h_2$ are monic and $\gcd(g_i,h_i)=1$ for $i=1,2$.
Since $f_i(P_\infty)\in\F_q$, the degrees of $g_i$ are less than or equal to those of  $h_i$  for $i=1,2$.
If $f_1(P_\infty)=f_2(P_\infty)=0$, then $P_\infty\not\in\supp(G_1\vee G_2)$ and  $\deg(g_i)<\deg(h_i)$ for $i=1,2$. Thus, $f_1-f_2=\frac{g_1h_2-g_2h_1}{h_1h_2}$ with $\deg(g_1h_2-g_2h_1)<\deg(h_1h_2)$.
This implies that $f_1-f_2\in\mL(G_1+ G_2-P_\infty)$.

If $f_1(P_\infty)=f_2(P_\infty)\in\F_q^*$, then the degrees of $g_i$ are equal to those of $h_i$ for $i=1,2$.
Moreover, the leading coefficients of $g_1(x)$ and $g_2(x)$ are equal.
As $g_1h_2$ and $g_2h_1$ have the same degrees and leading coefficients, we have $\deg(g_1h_2-g_2h_1)<\deg(g_1h_2)=\deg(h_1h_2)$. This implies that $(f_1-f_2)(P_\infty)=\frac{g_1h_2-g_2h_1}{h_1h_2}(P_\infty)=0$. By Lemma \ref{lem:2.2} and Case $1$, we have $f_1-f_2\in\mL(G_1+G_2-P_\infty)$.

{\bf Case 3: } If $f_1(P_\Ga)=f_2(P_\Ga)=\infty$, then by Lemma \ref{lem:2.1}, we have $P_\Ga \in\supp(G_1)\cap\supp( G_2)$, i.e., $P_\Ga\in\supp(G_1\wedge G_2)$.
By the identity $G_1\vee G_2=G_1+G_2-G_1\wedge G_2$, we have $G_1\vee G_2\le G_1+G_2-P_\Ga$. Since $f_1-f_2$ belongs to $\mL(G_1\vee G_2)$, it follows that  $f_1-f_2\in\mL(G_1+ G_2-P_\Ga)$.
\end{proof}

\subsection{Codes}
We denote a $q$-ary $(n,M,d)$ code as a code of length $n$, size $M$ and minimum distance $d$. The reader may refer to \cite{LX04,MS77,St09} for more details on coding theory.
There are some well-known bounds showing the restriction on the parameters of $n, M, d$ and $q$. One of the upper bound is the Singleton bound (see {\cite[Theorem 5.4.1]{LX04}}).

\begin{lemma}
For any integer $q>1$, any positive integer $n$ and any integer $d$ such that $1\le d\le n$, we have
\[M\le q^{n-d+1}.\]
\end{lemma}
A linear code achieving this bound is called a maximum distance separable (MDS) code.
Let $P_1,P_2,\cdots,P_n$ be the $n$ pairwise distinct places of degree one of the rational function field $F$ and $D=\sum_{i=1}^{n}P_i$ for $n\le q+1$. Let $G$ be a divisor of $F$ such that $0\le \deg(G)\le n-2$ and $\text{supp}(G)\cup \text{supp}(D)=\emptyset$. Then the rational algebraic geometry code $C_{\mL}(D,G)$ defined by
\[C_{\mL}(D,G):=\{(f(P_1),f(P_2),\cdots,f(P_n))|f\in \mL(G)\}\] is an $[n, \deg(G)+1, n-\deg(G)]$ MDS code over $\F_q$ \cite[Proposition 2.3.2]{St09}.

Due to rich algebraic structures of rational function fields over finite fields, various techniques have been employed to construct good codes from rational function fields  \cite{J15, J16,JK17,JMX17,JX15,TB14,X02}.
In particular, we will construct $(q+1)$-ary $(n,M,d)$ nonlinear codes via rational function fields in this paper. Hence, the size of our codes is upper bounded by:
\[M\le (q+1)^{n-d+1}.\]

\section{A new construction of nonlinear codes}\label{sec:3}
Let $q$ be a prime power. Let $\F_q=\{\Ga_1, \Ga_2, \cdots, \Ga_q\}$ be the finite field with $q$ elements.
Denote by  $\Sigma$ the set $\F_q\cup \{\infty\}$. The size of $\Sigma$ is $|\Sigma|=q+1$.
In this section, we will propose a construction of $(q+1)$-ary nonlinear codes over the code alphabet $\Sigma$ via the rational function field.

Our construction of codes is given explicitly as follows.
Let $F/\F_q$ be the rational function field and let $m$ be a positive integer.
Firstly, we choose a suitable set of rational functions which is a finite union of Riemann-Roch spaces
\[\mL_m:=\bigcup_{G\ge 0,\deg(G)\le m}\mL(G),\]
where $G$ runs through all the effective divisors of $F$ with degree $\le m$.
Then we consider an evaluation map $\phi: \mL_m\rightarrow \Sigma^{q+1}$ defined by
\[\phi(f):=(f(P_{\Ga_1}), f(P_{\Ga_2}), \cdots, f(P_{\Ga_q}), f(P_{\infty}))\in \Sigma^{q+1}.\]
Our code $C_m$ is constructed as the union of the image of $\phi$ and $\{(\infty, \infty, \cdots, \infty)\}$, that is to say,
\[C_m:=\{(f(P_{\Ga_1}), f(P_{\Ga_2}), \cdots, f(P_{\Ga_q}), f(P_{\infty})):\; f\in \mL_m\} \cup \{(\infty, \infty, \cdots, \infty)\}.\]
Our construction of nonlinear codes is different from the rational algebraic geometry codes in the sense that evaluations of rational functions are allowed at pole places as well.
This technique has been employed to construct $(q+1)$-ary nonlinear codes in \cite{SX05, X11}.
In the following, we will show that the explicitly constructed code $C_m$ has reasonable parameters with length $n=q+1$, size $M=q^{2m+1}+q^{2m}-2q^m+2$ and minimum distance $d=q+1-2m$.

In order to determine the size of $C_m$, we need to count the exact number of different rational functions in $\mL_m$.
In fact, it is a union of finitely many vector spaces over $\F_q$. Thus, it may not be a vector space over $\F_q$ anymore.
However, it is not difficult to see that $\mL_m$ can be characterized as follows.

\begin{lemma}\label{lem:3.1}
One has
\[ \mL_m=\left\{\frac{g(x)}{h(x)}:\;  g(x)\in  \F_q[x], h(x)\in \F_q[x]\setminus \{0\} \text{ with } \deg g(x)\le m, \deg h(x)\le m\right\}. \]
\end{lemma}

Now it is easy to obtain the exact number of different rational functions in $\mL_m$.

\begin{lemma}\label{lem:3.2}
The cardinality of $\mL_m$ is $q^{2m+1}+q^{2m}-2q^m+1$.
\end{lemma}
\begin{proof}
Consider two subsets of $\mL_m$ defined by
\[S_1=\left\{\frac{b_mx^m+\cdots+b_1x+b_0}{x^m+a_{m-1}x^{m-1}+\cdots+a_0} \in \mL_m:\; a_i, b_j\in \F_q, \text{ for } 0\le i\le m-1, 0\le j\le m\right\}\]
and
\[S_2= \left\{\frac{b_mx^m+\cdots+b_1x+b_0}{h(x)}\in \mL_m:\; b_m\in \F_q^*, h(x) \text{ is monic and } \deg h(x)<m\right\}. \]
We claim that $\mL_m$ is the disjoint union of $S_1$ and $S_2$. It is easy to see that $S_1$ and $S_2$ are disjoint.
It is sufficient to prove that $\mL_m\subseteq S_1\cup S_2$.
For any nonzero $z\in \mL_m$, we can write $z$ in the following form
$$z=\frac{g(x)}{h(x)}=\frac{b_mx^m+b_{m-1}x^{m-1}+\cdots+b_0}{a_mx^m+a_{m-1}x^{m-1}+\cdots+a_0}.$$
If $\deg g(x)\le \deg h(x)=k$, then $$z=\frac{g(x)}{h(x)}=\frac{a_k^{-1}g(x)\cdot x^{m-\deg h(x)}}{a_k^{-1}h(x)\cdot x^{m-\deg h(x)}}\in S_1.$$
If $\deg g(x)> \deg h(x)=k$, then $$z=\frac{g(x)}{h(x)}=\frac{a_k^{-1}g(x)\cdot x^{m-\deg g(x)}}{a_k^{-1}h(x)\cdot x^{m-\deg g(x)}}\in S_2.$$
Hence, $\mL_m$ is the disjoint union of $S_1$ and $S_2$. It follows that the number of distinct rational functions of $\mL_m$ is
\begin{align*}
|\mL_m| &=|S_1|+|S_2|\\ &= (q^{m+1}-1)\cdot q^m+1+(q-1)q^{m}\cdot (q^{m-1}+q^{m-2}+\cdots+1)\\ &=q^{2m+1}+q^{2m}-2q^m+1.
\end{align*}
This completes the proof.
\end{proof}

Now we can determine the parameters of our codes $C_m$.

\begin{prop}\label{prop:3.3}
Let $q$ be a prime power and let $m$ be a positive integer with $m\le q/2$.
Then the code $C_m$ is a $(q+1)$-ary $(n,M, d)$-code with length $n=q+1$, size $M=q^{2m+1}+q^{2m}-2q^m+2$ and minimum distance
\[d\ge q+1-2m.\]
\end{prop}
\begin{proof}
The length of the code $C_m$ is clearly $q+1$.
For a codeword $\bc=\phi(f)\in \Im(\phi(\mL_m))$, the number of poles of any rational function $f\in \mL_m$ is at most $m$, i.e., there are at most $m$ positions equal to $\infty$ in the codeword $\bc$.
Thus, the Hamming distance between $\bc$ and $(\infty, \infty, \cdots, \infty)$ is at least $q+1-m$.

Now let $\phi(f)$ and $\phi(g)$ be two distinct codewords in $\Im(\phi(\mL_m))$ with $f\neq g\in \mL_m$. Let $w$ be the Hamming distance between $\phi(f)$ and $\phi(g)$. Then there exists a subset $S$ of $\Sigma$ of size $q+1-w$ such that $f(P_{\Ga})=g(P_{\Ga})$ for all $\Ga\in S$. By Lemma \ref{lem:2.3}, we have $0\neq f-g\in\mL\left(G_1+G_2-\sum_{\Ga\in S}P_\Ga\right)$, where $G_1=(f)_\infty$ and $G_2=(g)_\infty$. Combining $\deg(G_i)\le m$ for $i=1,2$ and $\deg\left(G_1+G_2-\sum_{\Ga\in S}P_\Ga\right)\ge 0$, we obtain $|S|\le \deg(G_1)+\deg(G_2)\le 2m$, i.e., $w\ge q+1-2m$.
Hence, the minimum distance of $C_m$ is at least $q+1-2m$.

If $m\le q/2$, then the minimum distance of $C_m$ is larger than $0$. Hence, the evaluation map $\phi$ is injective and the size of $C_m$ is $| \mL_m|+1=q^{2m+1}+q^{2m}-2q^m+2$.
This completes the proof.
\end{proof}

Furthermore, we can show that the minimum distance of $C_m$ is exactly equal to $q+1-2m$. Firstly, let us prove an inequality related to the size of $C_m$.

\begin{lemma}\label{lem:3.4}
Let $q$ be a prime power and let $m$ be a positive integer with $m\le q/2$.  Then we have
\[q^{2m+1}+q^{2m}-2q^m+2>(q+1)^{2m}.\]
\end{lemma}
\begin{proof}
If $q=2$ and $m=1$, then we have $2^3+2^2-2\cdot 2+2=10>3^{2}$.
If $q\ge 3$ and $2m\le q$, then we have
\begin{eqnarray*}
q^{2m+1}+q^{2m}-2q^m+2 &>& 3 q^{2m} \\
&\ge & \left(1+\frac{1}{q}\right)^q \cdot q^{2m}
\ge  \left(1+\frac{1}{q}\right)^{2m} \cdot q^{2m}= (q+1)^{2m}.
\end{eqnarray*}
The second inequality follows from the fact that the infinite sequence $$\left\{\Big{(}1+\frac{1}{k}\Big{)}^k\right\}_{k=1}^{\infty}$$ is strictly increasing and upper bounded by the natural logarithm base $e=2.718\cdots$.
\end{proof}

Now we can show the main result of this paper.

\begin{theorem}\label{thm:3.5}
Let $q$ be a prime power and let $m$ be an integer such that $m\le q/2$.
The code $C_m$ defined by
\[C_m:=\{(f(P_{\Ga_1}), f(P_{\Ga_2}), \cdots, f(P_{\Ga_q}), f(P_{\infty})):\; f\in \mL_m\} \cup \{(\infty, \infty, \cdots, \infty)\}.\]
is a $(q+1)$-ary $(q+1, q^{2m+1}+q^{2m}-2q^m+2, q+1-2m)$ nonlinear code over $\Sigma=\F_q\cup \{\infty\}$.
\end{theorem}
\begin{proof}
Suppose that $d\ge q+2-2m$. Delete the first $d-1$ coordinates of each codeword of $C_m$. Then the remaining parts are still distinct codewords of length $n-d+1$.
The maximum number of codewords of length $n-d+1$ is $(q+1)^{n-d+1}\le (q+1)^{q+1-(q+2-2m)+1}=(q+1)^{2m}$.
As we know $q^{2m+1}+q^{2m}-2q^m+2> (q+1)^{2m}$ from Lemma \ref{lem:3.4}, the minimum distance $d$ of $C_m$ is exactly $q+1-2m$.
The remaining follows from Proposition \ref{prop:3.3}   immediately.
\end{proof}

\begin{rmk}
\begin{itemize}
\item[(1)]The code $C_m$ constructed in Theorem \ref{thm:3.5} is a $(q+1)$-ary $(q+1, M)$-code  with $M= q^{2m+1}+q^{2m}-2q^m+2$ achieving
the largest possible minimum distance. Indeed, from the Singleton bound, every $(q+1)$-ary $(q+1, M,d)$-code must obey $M=q^{2m+1}+q^{2m}-2q^m+2\le (q+1)^{q+2-d}$, i.e., $d\le q+1-2m$.
\item[(2)] If we consider the $(q+1)$-ary code  obtained from a $q$-ary $[q+1,2m+1,q+1-2m]$ MDS code  via code alphabet extension \cite{LX04}, then we get a $(q+1)$-ary code with parameters $(q+1,q^{2m+1},q+1-2m)$. This code has parameters worse than our code in this paper. For instance, a $9$-ary $[10, 5, 6]$ MDS code gives a $10$-ary $(10,59049,6)$ code. This is worse than our $10$-ary $(10,65450,6)$ code.
    \item[(3)] If $q+2$ is a prime power and we consider the $(q+1)$-ary code  obtained from a $(q+2)$-ary $[q+1,2m+1,q+1-2m]$ MDS code  via code alphabet restriction \cite{LX04}, then we get a $(q+1)$-ary code with parameters $(q+1,M,q+1-2m)$, where $M=\left\lceil\left(\frac{q+1}{q+2}\right)^{q+1}(q+2)^{2m+1}\right\rceil$. The parameters of the code  are worse than that of our code in this paper for large $q$. For instance, a $11$-ary $[10, 5, 6]$ MDS code gives a $10$-ary $(10,61843,6)$ code. Apparently this is not as good as  the $10$-ary $(10,65450,6)$ code from our construction. We will provide the details for this case in the next section.
\item[(4)]If $q+1$ is a prime power, then there exists a  $[q+1, 1+2m, q+1-2m]$ MDS linear code over $\F_{q+1}$ for each $1\le m\le q/2$ that have better parameters than the codes given in this paper. Hence, it only makes sense to consider the case where $q+1$ is not a prime power  such as $q=5, 9, 11, 13$, etc.
\end{itemize}
\end{rmk}

Take $m=1,2$, we derive the following corollaries.
\begin{cor}
The code $C_1$ is a $(q+1, q^3+q^2-2q+2, q-1)$-code over $\F_q\cup \{\infty\}$.
\end{cor}
It is easy to see that the size of $C_1$ is very close to the Singleton bound $(q+1)^3$.

\begin{cor}
For $q\ge 4$, the code $C_2$ is a $(q+1, q^5+q^4-2q^2+2, q-3)$-code over $\F_q\cup \{\infty\}$.
\end{cor}

\section{Numerical results and comparison}\label{sec:4}
In this section, we provide numerical examples from our construction in Section \ref{sec:3} and compare our bound with other $(q+1)$-ary nonlinear codes with length $q+1$.
First we list some of the nonlinear codes obtained from our construction for $q=5,9,11,13$ in the table \ref{table:4.1}. One can see that our codes have good parameters.

\begin{table}[h]
\caption{Nonlinear codes over $\F_q\cup \{\infty\}$}
\label{table:4.1}\vskip4pt
\begin{tabular}{||c|c|c|c|c|c|c||}
\hline
\hline m     & q=5 & q=9 & q=11 &q=13              \\
\hline 1  & (6, 142, 4) &(10, 794, 8)        &(12, 1432, 10)&(14, 2342, 12) \\
\hline 2  &(6, 3702, 2) &(10, 65450, 6)       &(12, 175452, 8)&(14, 399518, 10) \\
\hline 3  &             &(10, 5312954, 4)     &(12, 21256072, 6)&(14, 67570934, 8) \\
\hline 4  &             &  (10, 430454090, 2)  &(12,  2572277292, 4)&(14,  11420172974, 6) \\
\hline5 &               &                     &(12, 311248773112, 2)&(14,1930018143302, 4)\\
\hline6 &&&&                                      (14, 326173182061118, 2) \\
\hline \hline
\end{tabular}
\end{table}

In fact, most of the codes are constructed over finite fields in the literature. It is not easy to find $(q+1)$-ary codes with length $q+1$.
Luckily, Xing presented a construction of nonlinear codes over an arbitrary alphabet size from residue rings of polynomials in \cite{X02}.
Let $q$ be a prime power such that $q+2$ is a prime. It follows from \cite[Theorem 2.1]{X02} that there exists a $(q+1)$-ary $(q+1, M, \ge d)$-code with
 \[M\ge \frac{(q+1)^{q+1}}{(q+2)^{d-1}}\] for any positive integer $d$ satisfying $0<d<q+2$.
In particular, if $d=q+1-2m$, then
\begin{equation}\label{eq1}
M\ge \frac{(q+1)^{q+1}}{(q+2)^{q-2m}}.\end{equation}

Now we can compare our result with the bound given in \cite{X02}.

\begin{ex} Let $q=9$.
\begin{itemize}
\item[(1)] If $m=1$, then we have $q^{2m+1}+q^{2m}-2q^m+2=794$. However, \cite{X02} gives  $$\frac{(q+1)^{q+1}}{(q+2)^{q-2m}}=\frac{10^{10}}{11^7}<514.$$
\item[(2)] If $m=2$, then we have $q^{2m+1}+q^{2m}-2q^m+2=65450$. However, \cite{X02} gives  $$\frac{(q+1)^{q+1}}{(q+2)^{q-2m}}=\frac{10^{10}}{11^5}<62093.$$
\end{itemize}
\end{ex}

\begin{ex} Let $q=27$.
\begin{itemize}
\item[(1)] If $m=1$, then we have $q^{2m+1}+q^{2m}-2q^m+2=20360$. But the bound (\ref{eq1}) from \cite{X02}
shows 
$$\frac{(q+1)^{q+1}}{(q+2)^{q-2m}}=\frac{28^{28}}{29^{25}}<9131.$$
\item[(2)] If $m=2$, then we have $q^{2m+1}+q^{2m}-2q^m+2=14878892$. But the bound (\ref{eq1}) from \cite{X02}
shows
$$\frac{28^{28}}{29^{23}}=\frac{(q+1)^{q+1}}{(q+2)^{q-2m}}=\frac{28^{28}}{29^{23}}<7678404.$$
\item[(3)] If $m=3$, then we have $q^{2m+1}+q^{2m}-2q^m+2=10847734328$. But the bound (\ref{eq1}) from \cite{X02}
shows
$$\frac{(q+1)^{q+1}}{(q+2)^{q-2m}}=\frac{28^{28}}{29^{21}}<6457537275.$$

\end{itemize}
\end{ex}

From the above examples, we can see that our bound is better than the bound given in \cite{X02} when $q$ is sufficiently large compared with $m$.
In particular, we can show the following result.

\begin{prop}
Assume that $m$ is fixed, then we have
$$q^{2m+1}+q^{2m}-2q^m+2>\frac{(q+1)^{q+1}}{(q+2)^{q-2m}}$$ for sufficiently large $q$.
\end{prop}
\begin{proof}
The above inequality is equivalent to $$\frac{(q+2)^{q-2m}(q^{2m+1}+q^{2m}-2q^m+2)}{(q+1)^{q+1}}>1.$$
It is easy to verify that
\begin{eqnarray*}
\ln \frac{(q+2)^{q-2m}(q^{2m+1}+q^{2m}-2q^m+2)}{(q+1)^{q+1}}&=&(q-2m)\ln \frac{q+2}{q+1} + \ln \frac{q^{2m+1}+q^{2m}-2q^m+2}{(q+1)^{2m+1}}\\
&\rightarrow & \frac{q-4m}{q+1},
\end{eqnarray*}
when $m$ is fixed and $q$ approaches to infinity. Hence, this proposition follows immediately.
\end{proof}

\end{document}